\documentclass[11pt]{article}
\usepackage[a4paper, margin=1in]{geometry}
\usepackage{amsmath,amssymb,amsthm}
\usepackage{xcolor}
\usepackage{graphicx}
\usepackage{hyperref}
\usepackage{enumitem}
\usepackage{thmtools}
\usepackage[capitalize,nameinlink]{cleveref}
\usepackage{crossreftools}
\usepackage{dsfont}
\usepackage{authblk}

\newtheorem{theorem}{Theorem}
\newtheorem{lemma}[theorem]{Lemma}

\newtheorem{observation}[theorem]{Observation}
\newtheorem{definition}[theorem]{Definition}

\title{Maximum Cardinality Matching in $n^{2+o(1)}$ Time}
\author{George Z. Li}
\author{Jason Li}
\author{Junkai Zhang}
\affil{Carnegie Mellon University $\{\texttt{gzli}, \texttt{jmli},\texttt{junkaizh}\}$\texttt{@cs.cmu.edu}}
\date{}

\begin{document}

\maketitle

\begin{abstract}
    We consider the maximum cardinality matching problem in general graphs, and give an $n^{2+o(1)}$-time algorithm for the problem. Our work builds on recent work on combinatorial algorithms for maximum flow and extends the techniques to skew-symmetric flow problems which capture general matchings as a special case.
\end{abstract}

\textbf{Concurrent Work:} OpenAI released an $m^{1+o(1)}$-time algorithm for the same problem on October 6, 2026. We first obtained our result the day before this announcement and the techniques seem orthogonal. We believe our results are independently interesting, giving a simpler combinatorial algorithm for general matchings.

\textbf{Acknowledgments:}
Sanjeev Khanna and Junkai Song mentioned skew-symmetric flows as an approach towards general matchings algorithm when the first author was visiting NYU in Feburary 2026. Since then, the authors have discussed the approach in this paper with several humans (Aaron Bernstein and Isaac Hair) and several non-humans (GPT 5.5, 5.6, and 6).

\section{Introduction}

In recent years, there has been tremendous progress in designing faster algorithms for classical problems in graph algorithms. From maximum flow and min-cost flow~\cite{CKLPGS22, BBLST25, HJS26b} to minimum cut~\cite{ALPS23,HLRW24,FLSW26} to shortest paths~\cite{BNW22,DMMSY25,HLLZ26}, there have been numerous breakthrough on designing almost-linear time algorithms for many classical graph problems. One problem which remained untouched was the maximum matching problem in general graphs, where the classical $O(m\sqrt{n})$-time algorithm~\cite{MV80} remains the fastest in sparse graphs and matrix-multiplication based algorithms~\cite{MS04} remain fastest on dense graphs. Our main result is an improvement on dense graphs:

\begin{theorem}\label{thm:main}
        Given a unweighted general graph $G$ with $n$ vertices, we can compute a maximum matching of $G$ in $n^{2+o(1)}$ time with high probability.
\end{theorem}

In fact, our algorithm extends to the maximum integral symmetric flow problem on simple graphs, i.e., graphs with no parallel edges between any two vertices.

\begin{restatable}{theorem}{maximumisflow}\label{thm:maximum-is-flow}
        Given a simple, directed, skew-symmetric graph $G$ with $n$ vertices, and two symmetric vertices $s$ and $s'$, we can compute a maximum integer skew-symmetric flow from $s$ to $s'$ in $n^{2+o(1)}$ time with high probability.
\end{restatable}

\subsection{Technical Overview}

Given a directed, skew-symmetric graph, a maximum integral symmetric flow can be computed by iteratively finding a so-called \emph{regular} path in the current residual graph~\cite{GK96}. Finding a regular path can be implemented in linear time, so the algorithm takes $O(mF)$ time where $F$ is the value of the maximum integral symmetric flow. We have $F\le n$ on a simple graph, so the entire algorithm takes $O(mn)$ time.

At a high level, we speed up this algorithm by finding each regular path in \emph{sublinear} time, without looking at the entire residual graph. Our algorithm is inspired by a recent maximum flow algorithm~\cite{Li26}, which finds augmenting paths in a sparsifier of the current residual graph. Unlike the algorithm of~\cite{Li26}, ours does not compute a sparsifier upfront; instead, we alternate between searching the sampled graph for a regular path and sampling new edges to add to this graph.

For illustration, we first describe how our algorithm would find (not necessarily regular) augmenting paths in the standard maximum flow problem. The sampled graph is initially empty, and a search from the source vertex $s$ returns just $\{s\}$ as the set $S\subseteq V$ of vertices reachable from $s$. The algorithm tries to sample an outgoing edge from the current reachable set $S$, which would add a new vertex into $S$. Unfortunately, sampling edges from cuts in directed graphs is difficult in general. Instead, our algorithm uses a cut-sampling data structure on the underlying \emph{undirected} residual graph, which samples an edge on the (undirected) boundary of $S$. If a constant fraction of edges on the boundary of $S$ are outgoing from $S$, then the sampled edge is outgoing with constant probability. The issue is when most of the edges on the boundary are incoming to $S$. The algorithm of~\cite{Li26} ``toggles'' such ``unbalanced'' cuts by simultaneously increasing the sampling weight of all outgoing edges and decreasing the sampling weight of all incoming edges, a procedure handled implicitly by a dynamic data structure. Instead, our algorithm only modifies the weight of the sampled edge: if it is outgoing from $S$, its sampling weight is increased, and if it is incoming to $S$, its sampling weight is decreased. In a sense, our algorithm implements a ``stochastic'' version of the cut toggling operation in~\cite{Li26}.

We now return to the integral symmetric flow problem, which requires finding regular paths in the current residual graph. The main technical difficulty is that the absence of a regular path is not characterized by a cut, but a \emph{barrier} that decomposes the vertex set into arbitrarily many fragments, among other sets of vertices. Here, our algorithm takes inspiration from the shortest (weighted) regular path algorithm of~\cite{GK96}, which also alternates between searching and adding new edges to the current search graph. We assign \emph{charges} to the edges according to the barrier, such that edges with negative charge bypass the barrier similar to how the outgoing edges from $S$ bypass the cut. The charges are almost exactly the adjustment values in the shortest regular path algorithm of~\cite{GK96}, except that for technical reasons, we ignore the contribution of large fragments to the charges. When there are $K$ remaining augmenting path iterations, we define large fragments to be those with more than $\frac{2n}K$ vertices, of which there are at most $K/2$. Our algorithm targets a running time of $n^{2+o(1)}/K$ for the next augmenting path, which adds up to $n^{2+o(1)}$ time overall.

\section{Preliminaries}

As mentioned before, we consider the more general problem of maximum integral-symmetric flows, introduced in~\cite{GK04}, which captures the general matching and maximum flow problems as special cases. In this section, we will introduce all notations, definitions, and basic results we later use for graphs and skew-symmetric flows.

\subsection{Integral Skew-Symmetric Flows}\label{sec:is-flows}

\begin{definition}[Skew-symmetric graphs]\label{def:skew-symmetric-graph}
    A directed multigraph $G=(V,E)$ with positive integer edge capacities $c_e$ is called \emph{skew-symmetric} if it has a mapping $\sigma$ that maps vertices and edges in $G$ to their symmetric counterparts in $G$. This mapping satisfies the following conditions: 
    \begin{enumerate}[label=\roman*.]
        \item $\sigma(x)\ne x$ and $\sigma(\sigma(x))=x$ for every vertex and edge $x$
        \item if $e=(u,v)$ is an edge, then its symmetric edge $\sigma(e)$ is an edge from $\sigma(v)$ to $\sigma(u)$, and they have the same capacity, i.e., $c_e = c_{\sigma(e)}$
    \end{enumerate}
\end{definition}

For simplicity, we use $v'$ to denote the symmetric vertex $\sigma(v)$ of a vertex $v$, and $e'$ to denote the symmetric edge $\sigma(e)$ of an edge $e$. The symmetric set of a vertex set $S$ is defined as $S'=\{\sigma(v):v\in S\}$, and the symmetric path of a path $P=(e_1,\ldots, e_k)$ is defined as $P'=(e_k',\ldots,e_1')$.

\begin{definition}[Maximum integral symmetric flow (IS-Flow) problem]\label{def:is-flow}
    Given a skew-symmetric graph $G=(V,E)$ with a pair of symmetric vertices $s$ and $s'$, the \emph{maximum integral symmetric flow problem} asks to find the maximum flow from $s$ to $s'$ such that the flow $f_e$ on edges is integral and symmetric ($f_e = f_{e'}$).
\end{definition}

The general matching problem can be formulated as a maximum integral symmetric flow problem on the following ``matching graph'' construction.

\begin{definition}\label{def:matching-graph}
        For a graph $G=(V,E)$, the \emph{matching graph} $M_G$ is a skew-symmetric graph constructed as follows:
        \begin{itemize}
            \item For each vertex $v\in V$, create a vertex $v$ with its symmetric counterpart $v'$.
            \item For each edge $(u,v)\in E$, create an edge $(u,v')$ with its symmetric edge $(v,u')$. Both edges have capacity $1$.
            \item Add a source vertex $s$ with its symmetric counterpart $s'$. For every vertex $v\in V$, add edge $(s,v)$ with its symmetric edge $(v',s')$. Both edges have capacity $1$.
        \end{itemize}
\end{definition}

\begin{observation}\label{obs:matching-graph}
        Any general matching in the original graph $G$ corresponds to an integral symmetric flow in the matching graph $M_G$ whose flow value is twice the matching size. Conversely, any integral symmetric flow in $M_G$ corresponds to a general matching in $G$.
\end{observation}

\begin{proof}
        For every matching edge $(u,v)$ in a matching, we can map it to the following pair of integral symmetric paths in the matching graph $M_G$: $s\to u\to v'\to s'$ and $s\to v\to u'\to s'$. For the other direction, we can always decompose an integral symmetric flow in $M_G$ into such pairs of paths, which gives a matching in $G$.
\end{proof}

\subsection{Basic (Skew-Symmetric) Flow Theory}
\cite{GK96,GK04} introduced a theory of integral skew-symmetric flows, defining regular paths and barriers on skew-symmetric graphs as analogous objects for augmenting paths and cuts in standard flows. We focus on the unit capacity case for simplicity, as our matching graph $M_G$ satisfies this condition.

\begin{definition}\label{def:regular-path}
    A path from $s$ to $s'$ is called a \emph{regular path} if and only if it does not contain both an edge $e$ and its symmetric pair $e'$ simultaneously for any edge $e\in E$.
\end{definition}

\begin{figure}[t]
        \centering
        \includegraphics[width=0.5\textwidth]{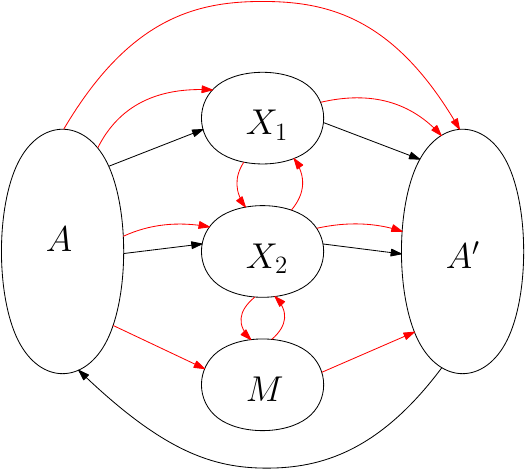}
        \caption{Example of a barrier in a skew-symmetric graph $H$. Black edges represent edges already in $H$. Red edges represent edges that do not exist in $H$, but can be added to $H$ to bypass the barrier.}
        \label{fig:cut-sampling}
\end{figure}

\begin{definition}\label{def:barrier}
    A \emph{barrier} from $s$ to $s'$ is a partition $(A,A',X_1,\ldots,X_k,M)$ of the vertex set $V$, where $A$ and $A'$ are symmetric sets and each fragment $X_i$ is self-symmetric, such that
    \begin{itemize}
        \item $s\in A$ (and symmetrically $s'\in A'$), and
        \item there are no edges from $A$ to $A'$, or from $A$ to $M$, and
        \item there are no edges between $X_i$ and $X_j$ for $i\ne j$, or between $X_i$ and $M$, and
        \item for every $i$, there is exactly one edge from $A$ to $X_i$.
    \end{itemize}
\end{definition}

\begin{theorem}\label{thm:regular-path-barrier}
    A regular path from $s$ to $s'$ exists if and only if no barrier from $s$ to $s'$ exists.
\end{theorem}

\begin{theorem}\label{thm:augmenting-regular-path}
        Given an skew-symmetric graph $G$ with edge capacity $1$ and an integral symmetric flow $f$ on $G$,
    \begin{itemize}
        \item The flow $f$ is maximum if and only if there is no regular path from $s$ to $s'$ in this residual graph.
        \item If there is a regular path from $s$ to $s'$, augmenting one unit along it and its symmetric path gives an integral symmetric flow of larger value.
    \end{itemize}
\end{theorem}

With these results, we can find maximum integral symmetric flows in any skew-symmetric graph by repeatedly finding regular paths and augmenting along them, until we hit a barrier on the residual graph. Our main result is that, instead of searching the graph in $O(m)$ time, we can find a regular path on the residual graph in amortized $n^{1+o(1)}$ time, which gives an $n^{2+o(1)}$ algorithm for maximum integral symmetric flows on matching graphs, hence yielding an $n^{2+o(1)}$ algorithm for general matchings.

\subsection{Graph Notation}
We now define the notations used in this paper.
Throughout the paper, we use $G$ to denote the skew-symmetric graph in the integral symmetric flow problem (except at the definition of matching graphs). We use $n$ and $m$ to denote the number of vertices and edges. We do not distinguish the size of the original graph and the skew-symmetric graph, as the matching graph construction in \Cref{def:matching-graph} only increases the number of vertices and edges by a constant factor. 

All graphs in this paper are simple and of unit capacity, though our algorithm also works in the same running time for graphs with at most a constant number of parallel edges between two vertices. For skew-symmetric graphs, we follow the simplified notation for symmetric objects in \Cref{sec:is-flows}, where we use $v'$ to denote the symmetric vertex $\sigma(v)$ of $v$, and $e'$ to denote the symmetric edge $\sigma(e)$ of $e$; the symmetric set of a vertex set $S$ is defined as $S'=\{\sigma(v):v\in S\}$, and the symmetric path of a path $P=(e_1,\ldots, e_k)$ is defined as $P'=(e_k',\ldots,e_1')$.
For a vertex subset $S$, we define the cut $\partial S$ as all edges between $S$ and $V\setminus S$, including both directions for directed graphs.

Throughout our algorithm, we maintain symmetric weights $w_e$ on all edges (including their corresponding reverse edges that may appear during augmentation) in the skew-symmetric graph $G$. For each pair $(e,e')$ of symmetric edges, we maintain a variable $x_e\in [1,U]$, which defines the weight of these two edges as $w_e=w_{e'}=\frac 1{x_e}$. We use $x_e$ to denote the variable associated with the edge $e$, so $x_e$ and $x_{e'}$ refer to the same variable. 

We call an edge \emph{active} if its capacity is $1$ in the current residual graph, and \emph{inactive} otherwise. The algorithm only maintains weights on active edges, and we reset the variable and weight whenever an edge becomes inactive after an augmentation.

\section{Finding Regular Paths in Residual Graphs}

In this section, we give our algorithm for finding the maximum integral symmetric flow in the residual skew-symmetric graph. This is the main technical result behind \Cref{thm:main}, as the general matching problem reduces to the maximum integral symmetric flow problem on the matching graph, constructed in \Cref{def:matching-graph}. 

\maximumisflow*

The core of this result is an algorithm that finds an integral symmetric flow on the residual graph of value at least half of the remaining flow value, which is stated in the following lemma.

\begin{lemma}\label{lem:finding-regular-path}
        There exists an algorithm that, given a skew-symmetric graph $G$ with unit capacity and a parameter $K$, assuming that the maximum integral symmetric flow of $G$ is at least $K$, finds an integral symmetric flow of value at least $\frac K2$ in $n^{2+o(1)}$ time with high probability.
\end{lemma}

Before stating our algorithm for \Cref{lem:finding-regular-path}, we first prove that it implies \Cref{thm:maximum-is-flow}, and thus \Cref{thm:main}.

\begin{proof}[Proof of \Cref{thm:maximum-is-flow}]
        We maintain an integral symmetric flow $f$ on the skew-symmetric graph $G$, and try to augment it in rounds. In each round, we try $K=2^i$ for each $i$ in increasing order, and run the algorithm in \Cref{lem:finding-regular-path} with the flow $f$ before this round. If some run of $K=2^j$ failed to get the desired flow in time, we augment $f$ with the flow returned by the last successful run, and proceed to the next round. 
        Assuming that all runs are correct, when the algorithm fails to compute an integral symmetric flow of value $\frac K2$, it means the remaining flow value is less than $K$. Since the last successful run augments $f$ by at least $\frac K4$, the algorithm reduces the remaining flow value by a constant factor in each round. By repeating for $O(\log n)$ rounds, we get the maximum integral symmetric flow in $G$.
\end{proof}

\subsection{Algorithm Overview}

We start by giving a high-level overview of our algorithm. Our algorithm maintains an integral symmetric flow $f$ and tries to augment $f$ by finding a regular path from $s$ to $s'$ in the residual graph, and then augmenting along it and its symmetric path. To avoid visiting all $m$ edges to find a regular path, we run a regular reachability algorithm on a subgraph $H$ of the residual graph. When the algorithm gets stuck in the subgraph, i.e., it finds a barrier, we add a pair of symmetric edges from the residual graph to $H$ to bypass this barrier, and continue the search. Each included edge pair helps the algorithm reach new vertices using regular paths, so we only add $O(n)$ edges to the subgraph $H$ before we find a regular path. 

When the regular reachability algorithm finds a barrier in the subgraph $H$, we try to find an edge in the residual graph that can bypass the barrier, as indicated by the red edges in \Cref{fig:cut-sampling}, and add it to the subgraph. We find such edges by sampling. For the barrier we encountered, we define a charge $d_e$ on every edge $e$ in the residual graph, such that we can efficiently sample an edge with non-zero charge with probability proportional to its weight $w_e$ times the absolute value of its charge $|d_e|$. Moreover, any negatively charged edge can bypass the barrier, and such an edge exists if the remaining flow is large enough. Finally, negatively charged edges can be sampled with constant probability if we have a ``balanced'' weighting of the residual graph, like the flow value of an integral symmetric flow. If the edge weights are balanced, we can find an edge that can bypass the barrier with high probability in few rounds of sampling. Otherwise, if the edge weights are unbalanced, we can modify the edge weights so they become more balanced. The detailed process of sampling and modifying weights is described in the next subsection, before we give the full algorithm.

\subsection{Barrier toggle operations}
Suppose that the algorithm finds a barrier in the subgraph $H$. Recall that a barrier is a partition $(A, A', X_1, \ldots, X_k, M)$ of the vertex set in the skew-symmetric graph such that $s\in A$, and
\begin{itemize}
        \item $X_1,\ldots, X_k, M$ are self-symmetric sets, and $A'$ is the symmetric set of $A$, and
        \item there are no edges from $A$ to $A'$ or from $A$ to $M$, and
        \item there are no edges between two different fragments $X_i$ and $X_j$, or between $X_i$ and $M$, and
        \item there is exactly one edge from $A$ to each fragment $X_i$.
\end{itemize}

And for this barrier, we can add one of the following edges (see \Cref{fig:cut-sampling} for an illustration) to $H$ to bypass the barrier:

\begin{itemize}
        \item edges from $A$ to a fragment $X_i$, excluding the existing one in the barrier, or
        \item edges between two different fragments $X_i$ and $X_j$, or
        \item edges from $A$ to $A'$, or
        \item edges from $A$ or $X_i$ to $M$.
\end{itemize}

\begin{figure}[t]
        \centering
        \includegraphics[width=0.5\textwidth]{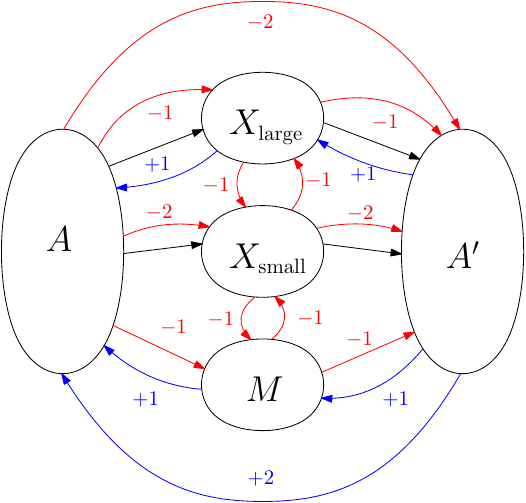}
        \caption{Illustration of the operation vector for a barrier. $X_{\textup{small}}$ is a fragment of size $\le \frac{2n}K$, and $X_{\textup{large}}$ is a fragment of size $>\frac {2n}K$. All edges with non-zero charge are drawn with edge weights showing their charge. Barrier edges (black edges) have charge zero.}
        \label{fig:cut-sampling-weight}
\end{figure}

For this barrier, we define an \emph{operation vector} $\overrightarrow{d}$ for toggling this barrier. It assigns a charge $d_e$ to every active symmetric edge pair $(e,e')$, defined as follows.

\begin{definition}\label{def:operation-vector}
        Consider a barrier $(A, A', X_1, \ldots, X_k, M,)$ and a parameter $K$. Let $e_i$ be the edge from $A$ to $X_i$ in the barrier. Let $\phi(u)$ be an auxiliary vertex potential such that
        \[\phi(u)=\mathds{1}[u\in A']-\mathds{1}[u\in A].\]
        We define the charge on an active edge $e=(u,v)$ as
        \[d_e=\phi(u)-\phi(v)+\sum_{i}\mathds 1\left[e=e_i\lor e=e_i'\right]-\sum_{X_i: |X_i|\le \frac {2n}K}\mathds 1\left[e\in \partial X_i\setminus \{e_i,e_i'\}\right].\]
        The operation vector $\overrightarrow{d}$ is the vector of charges $d_e$ for all active edges $e$ in the graph.
\end{definition}

This edge charge has the following properties. First, charges $d_e$ are symmetric and in the range $\{-2, -1, 0, 1, 2\}$. Second, negatively charged edges can be added to the subgraph $H$ to bypass the barrier. These properties are illustrated in \Cref{fig:cut-sampling-weight}. 
The parameter $K$ is the same as in \Cref{lem:finding-regular-path}, which is a lower bound on the remaining flow before any augmentation. We exclude large fragments with $|X_i|>\frac{2n} K$ to support sampling an edge with probability proportional to $|d_e|w_e$ efficiently, as we will show in \Cref{sec:sampling}. 

Suppose that we can sample an edge $e$ with probability proportional to $|d_e|w_e$ efficiently. If negatively charged edges constitute a constant fraction of the total weight of edges with non-zero charges, then we can find a negatively charged edge with high probability in a few rounds of sampling, which can be added to the subgraph $H$ to bypass the barrier. Otherwise, the weight is unbalanced, since the total weight of positively charged edges is much larger than the total weight of negatively charged edges. In this case, we can use the following barrier toggle operation to modify the edge weights. The operation below is stated in a ``stochastic'' form for efficiency. In expectation, it shifts the variable $x_e$ of each edge by a value proportional to its charge $d_e$, so negatively charged edges now have a smaller variable $x_e$ and thus larger weight $w_e$, while positively charged edges have a larger variable $x_e$ and thus smaller weight $w_e$. 

\begin{definition}[Barrier toggle operations]\label{def:barrier-toggle}
        For an operation vector $\overrightarrow{d}$ given by a barrier, we define the operation of toggling this barrier as follows: Sample an edge $e$ with probability proportional to $|d_e|w_e$. If the sampled edge has negative charge ($d_e<0$), we multiply $x_e$ by $\frac 12$. If the sampled edge has positive charge ($d_e>0$), we multiply $x_e$ by $\frac 32$.
\end{definition}

\subsection{The full algorithm}
Using the edge weights and barrier toggle operations defined above, we now describe the full algorithm for \Cref{lem:finding-regular-path}, which finds an integral symmetric flow of value $\frac K2$ on the residual graph, assuming that the remaining integral symmetric flow has value at least $K$. 

\paragraph{Algorithm.} Initialize $x_e=1$ for every edge pair $(e,e')$ in $G$. Repeat the following steps to find regular paths in the residual graph of $G$ until we increase the flow value by at least $\frac K2$, or the algorithm declares a failure.

\begin{enumerate}
        \item Initialize the subgraph $H$ to contain no edges.
        \item Run the regular reachability algorithm on $H$ until it finds a barrier or a regular path.
        \item If a barrier is found, apply the barrier toggle operation defined in \Cref{def:barrier-toggle} on this barrier to modify the edge weights.
        \item If the sampled variable $x_e$ drops below $1$, clip it back to $1$. If the sampled variable $x_e$ exceeds $U$, declare failure and terminate the algorithm. 
        \item If the sampled edge has a negative charge, add it and its symmetric edge to $H$ and resume the search from step 2. Otherwise, repeat this operation starting from step 3.
        \item If a regular path is found, augment along it and its symmetric path to update the flow. Remove all edges that become inactive from supporting data structures, then add all edges that become active to these data structures while setting their variable $x_e=1$. 
\end{enumerate}

For regular reachability algorithms, we use the algorithm in \cite{GK96}. This algorithm naturally supports resuming after adding edges to the subgraph, which we will discuss in~\Cref{app:resumable-rr} in more detail. The edge weight upper bound $U$ is polynomial in $n$, and we will choose it later in \Cref{sec:edge-variable-bound} so that the algorithm succeeds with high probability.

The remaining part of this section is to analyze the correctness of this algorithm, which completes the proof of \Cref{lem:finding-regular-path}. The next subsection bounds the total number of barrier toggle operations by $\tilde O(n^2)$, and the subsection after that shows that the algorithm succeeds with high probability. Finally, we show that the sampling step in the barrier toggle operation can be implemented efficiently, so the algorithm runs in $n^{2+o(1)}$ time.

\subsection{Energy bound of barrier toggle operations}
In this subsection, we show that the total number of barrier toggle operations is bounded by $\tilde O(n^2)$ for polynomially bounded $U$.
In each operation, we either sample a negatively charged edge, which allows the algorithm to reach new vertices using regular paths, or sample a positively charged edge, which multiplies the variable $x_e$ by $\frac 32$. The former case only occurs $O(n)$ times for finding a regular path, and $O(n^2)$ times for the whole algorithm. Assuming that $U$ is polynomially bounded, the latter case also only occurs $\tilde O(n^2)$ times before some variable $x_e$ exceeds $U$. 

Formally, we define the energy $\Phi$ as follows:
\[\Phi = \sum_{\text{active }(e,e')} \int_{x_e}^{2U}\frac 1t\ \text{d}t=\sum_{\text{active }(e,e')} \left(\ln (2U) - \ln x_e\right).\]

We first claim that the energy $\Phi$ is always non-negative. As long as the algorithm does not declare failure, we have $x_e\in [1,U]$. Since we multiply $x_e$ by at most $\frac 32$ in each barrier toggle operation, we have $x_e\le 2U$ after any operation, so $\Phi\ge 0$. 
Next, we consider the changes in energy $\Phi$ in the algorithm. There are three sources of energy increase: initialization, augmentation, and sampling a negatively charged edge during barrier toggle operations. As each edge contributes at most $\ln (2U)$ to the energy, initialization contributes at most $m\ln (2U)$ to the energy. Each augmentation contributes at most $n\ln (2U)$ to the energy. Sampling a negatively charged edge increases the energy by at most $\ln 2$. As every such edge allows the algorithm to reach new vertices using regular paths, there are at most $O(n)$ such operations in each augmentation. Since there are at most $O(n)$ augmentations, the total energy increase from all sources is $O(n^2\log (2U))$.
Whenever we sample a positively charged edge, the energy decreases by $\ln \frac 32$. Therefore, the total number of barrier toggle operations that sample a positively charged edge is at most $O(n^2\log (2U))$. Combined with the $O(n^2)$ operations that sample a negatively charged edge, the total number of barrier toggle operations is at most $O(n^2\log (2U))$.

\subsection{Preservation of maximum integral symmetric flows}\label{sec:edge-variable-bound}

In this subsection, we prove that, if we set $U$ to be a sufficiently large polynomial in $n$, then with high probability, the variable $x_e$ of any edge will not exceed $U$ before we find an integral symmetric flow of value at least $\frac K2$, assuming that the maximum integral symmetric flow has value at least $K$.

\begin{lemma}\label{lem:preserve-max-flow}
        Let $f^*$ be a maximum integral symmetric flow in $G$. Suppose that $|f^*|\ge K$. With probability at least $1-O\left(\frac{n^3\log (2U)}{\sqrt U}\right)$, all edge variables $x_e$ will not exceed $U$ throughout the algorithm in \Cref{lem:finding-regular-path}.
\end{lemma}

We prove this lemma in two steps. First, we use the lemma below to show that the total change of edges in the current residual flow $f^*-f$ is negative. Since the expected change of each edge variable $x_e$ is proportional to its charge $d_e$, this implies that the variables of edges in the support of $f^*-f$ are bounded with high probability. Next, we show that if the above holds, then the expected increase of each active edge variable $x_e$ is also bounded, which implies that all active edge variables are bounded with high probability.

\begin{lemma}\label{lem:regular-path-charge}
        Let $G$ be the current residual graph, and $\overrightarrow{d}$ be an operation vector based on a barrier on a subgraph of $G$. For an integral symmetric flow $f$ in $G$ of value at least $\frac K2$, the total charge on edges in the flow is negative, i.e., $\sum_{e} f_e d_e < 0$.
\end{lemma}

\begin{proof}
        Let the barrier be $(A, A', X_1, \ldots, X_k, M)$. Recall that the charge for an edge $e=(u,v)$ is defined as
        \[d_e=\phi(u)-\phi(v)+\sum_{i}\mathds 1\left[e=e_i\lor e=e_i'\right]-\sum_{X_i: |X_i|\le \frac {2n}K}\mathds 1\left[e\in \partial X_i\setminus \{e_i,e_i'\}\right]\]
        where $\phi(u)=\mathds{1}[u\in A']-\mathds{1}[u\in A]$.
        
        By the flow decomposition theorem of skew-symmetric flows (\cite{GK04}), we can split $f$ into a set of regular paths from $s$ to $s'$ or from $s'$ to $s$, with additional regular cycles that do not use both an edge and its symmetric edge simultaneously.
        We consider the contribution of vertex potential $\phi$, small fragment terms, and large fragment terms to the total charge separately.
        \begin{itemize}
                \item For any path from $s$ to $s'$, vertex potential $\phi$ contributes $\phi(s)-\phi(s')=-2$ to the total charge, and this path contributes $1$ to the flow value. For any path from $s'$ to $s$, vertex potential $\phi$ contributes $2$ to the total charge, and this path contributes $-1$ to the flow value. Vertex potential $\phi$ contributes $0$ to the total charge of a cycle. Therefore, the total contribution of vertex potential $\phi$ to the total charge is $-2|f|$.
                \item For each small fragment $X_i$, the term $\mathds 1\left[e=e_i\lor e=e_i'\right]-\mathds 1\left[e\in \partial X_i\setminus \{e_i,e_i'\}\right]$ contributes at most $0$ to any regular path or cycle, as a regular path or cycle cannot use symmetric edges $e_i$ and $e_i'$ simultaneously, and if it uses one of them, it must use another edge in $\partial X_i\setminus\{e_i,e_i'\}$ to enter or exit $X_i$. 
                \item For a large fragment $X_i$, the only term is $\mathds 1\left[e=e_i\lor e=e_i'\right]$, which contributes at most $2$ to the whole flow. As a large fragment has size strictly larger than $\frac{2n}K$, there are fewer than $\frac K2$ such fragments, so the total contribution here is less than $2\cdot \frac K2=K$.
        \end{itemize}
        Therefore, the total charge of edges in the flow is strictly less than $-2|f|+K\le 0$ for flows of value at least $\frac K2$.
\end{proof}

Using this lemma, we can prove \Cref{lem:preserve-max-flow}.

\begin{proof}[Proof of \Cref{lem:preserve-max-flow}]
        Let $f$ be the current integral symmetric flow found by the algorithm, so that $f^*-f$ is supported by the current residual graph. In particular, all edges in the support of $f^*-f$ are active.

Our proof proceeds by defining a super-martingale and applying the optional stopping theorem. We first track the quantity
        \[Q=\bigg(\sum_{\text{symmetric pair } (e,e')\in \operatorname{supp}(f^*-f)} x_e\bigg)+\Phi+cn^2\log (2U)-N-2An\log (2U)\]
    throughout the algorithm, where $\Phi$ is the total energy, $N$ is the number of barrier toggle operations where a negative charge edge is sampled, $A$ is the number of augmentations already made, and $c$ is a large enough constant so that $N+2An\log (2U)\le cn^2\log (2U)$ always holds.

The quantity $Q$ defines a stochastic process that is frozen when the algorithm terminates, possibly with failure. It is always non-negative since $x_e\ge 1$, $\Phi\ge 0$, and $cn^2\log (2U)\ge N+2An\log (2U)$. Moreover, the initial value of $Q$ is at most $(c+2)n^2\log (2U)$.

We now show that $Q$ is a super-martingale, i.e.,  the expected value of $Q$ does not increase throughout the algorithm. We analyze the expected change of $Q$ in each step below.

    \begin{itemize}
        \item For augmentations, we set $x_e=1$ for augmented edges, and then replace them by their reverse edges in the residual graph.
        The first step does not increase the sum of $x_e$ for edge pairs $(e,e')$ in $\operatorname{supp}(f^*-f)$. The second step increases it by at most $n$ as the algorithm changes at most $n$ edge pairs, and each of them has $x_e=1$.
        This also increases $\Phi$ by at most $n\log (2U)$ by removing some edges and adding at most $n$ reverse edge pairs with $x_e=1$. Therefore, the expected increase from the first two terms is at most $2n\log (2U)$, which is canceled by the last term $-2An\log (2U)$ decreasing by $2n\log (2U)$.
        \item For barrier toggle operations, we consider the expected change of each term in $Q$.
        For the energy term $\Phi$, if the operation samples a positive charge edge, then $\Phi$ decreases by $\ln \frac 32$, and $N$ is unchanged. If it samples a negative charge edge $e$, then $\Phi$ increases by $\ln 2$, and $N$ increases by $1$. Therefore, $\Phi-N$ does not increase.
        Next, we consider the sum of $x_e$ for edges $e\in \operatorname{supp}(f^*-f)$. In barrier toggle operations, the probability of sampling an edge $e$ is proportional to $|d_e|w_e=\frac {|d_e|}{x_e}$, and the change to this variable $x_e$ is $\frac12 \operatorname{sgn}(d_e)x_e$. So the expected change to a variable $x_e$ is proportional to its charge $d_e$. Since the remaining flow value is at least $\frac K2$, \Cref{lem:regular-path-charge} gives that the total charge of edges in $\operatorname{supp}(f^*-f)$ is negative, so the expected change to $\sum_{\text{symmetric pair} (e,e')\in \operatorname{supp}(f^*-f)} x_e$ is negative. Therefore, the expected value of $Q$ does not increase during barrier toggle operations.
                \item For the clipping step, the only change to variables is that a variable $x_e<1$ gets clipped back to $1$. If this edge is not in $\operatorname{supp}(f^*-f)$, then only $\Phi$ decreases, and $Q$ does not increase. Otherwise, $\Phi$ decreases by $-\ln x_e$, and $\sum_{\text{symmetric pair} (e,e')\in \operatorname{supp}(f^*-f)} x_e$ increases by $1-x_e\le -\ln x_e$, so $Q$ does not increase. 
    \end{itemize}

Define a bound $B\le U$ to be chosen later. Consider stopping the super-martingale $Q$ at the first time some edge $e\in \operatorname{supp}(f^*-f)$ has $x_e>B$, or when the algorithm terminates, possibly with failure. By the optional stopping theorem, the expected value of $Q$ at this stopping point is at most the initial value, which is at most $(c+2)n^2\log (2U)$. In the event some edge $e\in \operatorname{supp}(f^*-f)$ has $x_e>B$, we have $Q\ge B$. Since $Q$ is nonnegative, the probability that this event occurs is at most $\frac{(c+2)n^2\log (2U)}{B}$ by Markov's inequality.

Consider another stochastic process
\[ R=\bigg(\sum_{\text{symmetric pair } (e,e') \text{ in } H} x_e\bigg)+(cn^2\log(2U)-M)(B+1)n^2 \]
where $M\le cn^2\log(2U)$ is the total number of barrier toggle operations, except that we freeze this process the moment any edge $e\in\operatorname{supp}(f^*-f)$ has $x_e>B$. We now show that $R$ is also a super-martingale.

\begin{itemize}
\item For augmentations, we remove edges with variable $x_e\ge1$ and replace them by their reverse edges with variable $1$, so $R$ cannot increase.

\item For each barrier toggle operation, let $Z=\sum |d_e|w_e$ denote the total sample weight. Since \Cref{lem:regular-path-charge} implies that the total charge of edges in $\operatorname{supp}(f^*-f)$ is negative, some edge $e\in\operatorname{supp}(f^*-f)$ has nonzero charge. Assuming that the stochastic process is not frozen, we have $x_e<B$, so $Z\ge w_{e^*}\ge\frac1B$.

The probability of sampling each edge $e$ is $\frac{|d_e|w_e}{Z}$, and the change to this variable $x_e$ is $\frac12 \operatorname{sgn}(d_e)x_e$ before clipping. Therefore, $\mathbb {E} [x_e]$ increases by $\frac{d_e}{2Z}$ for each edge $e$ before clipping. 
    Together with $Z\ge\frac1B$ and $d_e\le2$, we conclude that $\mathbb{E} [x_e]$ increases by at most $B$ in each barrier toggle operation before clipping, and possibly by another $1$ after clipping. Therefore, the expected increase of $\sum x_e$ is at most $(B+1)m$, which is canceled by the $(cn^2\log(2U)-M)(B+1)n^2$ term decreasing by $(B+1)m$.
\end{itemize}

Consider stopping the super-martingale $R$ at the first time some edge $e$ has $x_e>U$, or when the algorithm terminates, possibly with failure. By the optional stopping theorem, the expected value of $R$ at this stopping point is at most the initial value, which is at most\linebreak $(c+1)\log(2U)\,(B+1)n^4$. In the event some edge $e$ has $x_e>U$ \emph{and} no edge $e\in \operatorname{supp}(f^*-f)$ has $x_e>B$ at a previous point in time (which would freeze the process), we have $Q\ge U$. Since $Q$ is nonnegative, the probability that this event occurs is at most $\frac{(c+1)\log(2U)\,(B+1)n^4}{U}$ by Markov's inequality.

We have bounded the probability that any edge $e\in \operatorname{supp}(f^*-f)$ has $x_e>B$ (at any point in time) by $O\left(\frac{n^2\log (2U)}{B}\right)$, and the probability that some edge $e$ has $x_e>U$ \emph{and} no edge $e\in \operatorname{supp}(f^*-f)$ has $x_e>B$ at a previous point in time by $O\left(\frac{Bn^4\log (2U)}{U}\right)$. The probability that some edge $e$ has $x_e>U$ is at most the sum of these two probabilities. By setting $B=\frac{\sqrt U}{n}$, the total failure probability is at most $O\left(\frac{n^3\log (2U)}{\sqrt U}\right)$. 
\end{proof}

\subsection{Efficient non-zero charge edge sampling}\label{sec:sampling}

Now we consider the sampling process in the barrier toggle operation. We will show that this operation can be reduced to sampling edges from a cut $\partial S$ in an undirected graph, with updates to the cut $S$ and the graph, as stated below.

\begin{lemma}\label{lem:sampling-from-cut}
        There is an algorithm that maintains an undirected graph $G$ with polynomially bounded edge weights and a cut $S\subseteq V$, using $m^{1+o(1)}$ preprocessing time, and supports the following operations in $m^{o(1)}$ amortized time:
        \begin{itemize}
                \item Add or remove an edge in the graph.
                \item Add or remove a vertex from the cut $S$.
                \item Sample an edge from the cut $\partial S$ with probability proportional to its weight.
        \end{itemize}
\end{lemma}

We first prove this lemma using expander techniques. Suppose that the graph is a $\phi$-expander, i.e., for any cut $S$, $w(\partial S)\ge \phi\min\{\operatorname{vol}(S), \operatorname{vol}(V\setminus S)\}$, where $\operatorname{vol}(S)=\sum_{v\in S} \operatorname{deg}(v)$ is the total degree of the vertices in $S$. In this simple case, we can sample an edge from the cut $\partial S$ as follows: First sample a vertex $v$ from $S$ ($V\setminus S$ if $\operatorname{vol}(S)\ge \operatorname{vol}(V\setminus S)$) with probability proportional to its degree $\operatorname{deg}(v)$, and then sample an edge adjacent to $v$; if this edge is not in the cut $\partial S$, we reject it and repeat the process. Each edge in the cut $\partial S$ is sampled with probability proportional to its weight, and the expected number of repeats is at most $\frac 1{\phi}$ because of the $\phi$-expansion property. We extend this idea to general graphs by maintaining a sequence of dynamic expander decompositions, similar to the one used for directed graph sparsification in \cite{Li26}, which fully decomposes the edges of the graph into expanders. We start by stating the dynamic expander decomposition on unweighted graphs that we will use in our algorithm.

\begin{definition}[$\phi$-expander]\label{def:expander}
    An undirected graph is a $\phi$-expander if for any cut $S\subseteq V$, $w(\partial S)\ge \phi\min\{\operatorname{vol}(S), \operatorname{vol}(V\setminus S)\}$,
\end{definition}

\begin{definition}[$\phi$-expander decomposition with slack $s$]\label{def:expander-decomposition}
    A $\phi$-expander decomposition of an unweighted graph $G$ with slack $s$ is a partition $\mathcal P=\{V_1,\ldots,V_k\}$ of the vertex set $V$ such that the induced subgraph $G[V_i]$ on each cluster $V_i$ is a $(\phi/s)$-expander, and the total number of edges between different clusters, denoted as $\partial\mathcal P$, is at most $\phi m$.
\end{definition}

\begin{theorem}[Theorem 7.2 of \cite{GRST21}]\label{thm:dynamic-expander-decomposition}
    Let $0<\phi<1$ be a parameter. There is a randomized algorithm that, given an unweighted graph $G$ undergoing edge insertions and deletions, maintains a $\phi$-expander decomposition $\mathcal P$ with slack $s=2^{O(\log^{1/2} n)}$, with the explicit set of edges $\partial\mathcal P$. The algorithm has update time $2^{O(\log^{1/2} n)}/\phi^2$ and amortized recourse $2^{O(\log^{1/2} n)}$ with respect to $\partial \mathcal P$, which means that the amortized changes to the set $\partial \mathcal P$ is at most $2^{O(\log^{1/2} n)}$ per update in $G$.
\end{theorem}

Using these results, we can prove \Cref{lem:sampling-from-cut}.

\begin{proof}[Proof of \Cref{lem:sampling-from-cut}]
    For every possible $i$, let $G_i$ be the subgraph of $G$ that contains all edges with weight in $(2^{i-1}, 2^i]$. For each $G_i$, similar to Theorem 10 of \cite{Li26}, we maintain a sequence of dynamic expander decompositions as follows: We first apply \Cref{thm:dynamic-expander-decomposition} and maintain a expander decomposition $\mathcal P_i^{(1)}$ of $G_i$ with $\phi=2^{-\log^{3/4} n}$, then recursively perform this operation on the remaining edges $\partial \mathcal P_i^{(1)}$ between different expanders to get $\mathcal P_i^{(2)}$, and so on. As the number of edges is reduced by a factor of $2^{\log^{3/4} n}$ for each level of recursion, the total number of levels is $O(\log^{1/4} n)$. As the amortized recourse of each level is $2^{O(\log^{1/2} n)}$, the overall recourse is $2^{O(\log^{3/4} n)}$, and the total update time is also $2^{O(\log^{3/4} n)}$ per update. 

    This process splits the edges of $G_i$ into many $2^{-O(\log^{3/4} n)}$-expanders, where each edge is in exactly one expander, and each vertex is in at most $O(\log^{1/4} n)$ expanders. For each vertex $v$ in each expander $A$, we maintain its degree $\operatorname{deg}_A(v)$ in $A$, and the set of its adjacent edges in $A$. For each expander $A$ with vertex support $U\subseteq V$, we maintain its intersection $U\cap S$ with the cut $S$, along with each vertex's degree $\operatorname{deg}_A(v)$ for $v\in U\cap S$, and we also maintain the same information for the complement $U\setminus S$. Finally, for each expander $A$, supposing that it comes from graph $G_j$, we set its weight to be $2^j\min\{\operatorname{vol}_A(U\cap S), \operatorname{vol}_A(U\setminus S)\}$. All of this information can be maintained for cut updates, as each vertex is in at most $O(\log^{1/4} n)$ expanders, and each edge is in exactly one expander.

        For sampling, we first sample an expander $A$ (of vertex support $U$) with probability proportional to its weight $2^j\min\{\operatorname{vol}_A(U\cap S), \operatorname{vol}_A(U\setminus S)\}$, where $j$ is the index of the graph $G_j$ that contains $A$. Next, supposing that $B$ is the smaller side of the cut $U\cap S$ and $U\setminus S$, we sample a vertex $v$ from $B$ with probability proportional to its degree $\operatorname{deg}_A(v)$ in $A$, and then sample an edge adjacent to $v$ in $A$. If this edge is not in the cut $\partial S$, we reject it and repeat; If this edge is in $\partial S$, we also reject and repeat with probability $1-\frac{w_e}{2^j}\le\frac12$. The expanding property of $A$ guarantees that we choose an edge in $\partial S$ with probability at least $2^{-O(\log^{3/4} n)}$. The probability that an edge $e$ in the cut $\partial S$ is sampled is proportional to $2^j$, where $j$ is the index of the graph $G_j$ that contains $e$, and thus the final step guarantees that the probability of sampling $e$ is proportional to its weight $w_e$. The expected number of repeats is at most $\frac 2{2^{-O(\log^{3/4} n)}}=n^{o(1)}$.
\end{proof}

We now finish the task of sampling an edge $e$ with non-zero charge $d_e$ using this lemma by reducing our task to sampling edges from a cut $\partial S$ in an undirected graph.
Recall that for a barrier $(A, A', X_1, \ldots, X_k, M)$, we define the edge charge $d_e$ for every edge as
\[d_e=\phi(u)-\phi(v)+\sum_{i}\mathds 1\left[e=e_i\lor e=e_i'\right]-\sum_{X_i: |X_i|\le \frac {2n}K}\mathds 1\left[e\in \partial X_i\setminus \{e_i,e_i'\}\right]\]
where $\phi(u)=\mathds{1}[u\in A']-\mathds{1}[u\in A]$. Our goal is to sample an edge $e$ with probability proportional to $|d_e|w_e$. 

Our first step is to write the charge $d_e$ for edge $e=(u,v)$ as the difference of vertex heights $h(u)-h(v)$ for some function $h$ on vertices, except for barrier edges and edges in a small fragment $X_i$, as follows.

\begin{lemma}\label{lem:edge-charge-to-vertex-height}
        Let $G$ be the current residual graph, and let $(A, A', X_1, \ldots, X_k, M)$ be the barrier found by the algorithm that is used to define the operation vector $\overrightarrow{d}$. Consider an auxiliary bipartite graph $G_1$ as follows: For each vertex $u$ in $G$, add two copies $u^L$ and $u^R$ into $G_1$; for each edge $e=(u,v)$ in $G$, add an edge $(u^L,v^R)$ into $G_1$. Define height $h$ on vertices in $G_1$ as follows:
        \begin{align*}
                h(u^L)&=\phi(u)-\mathds 1\left[u\in X_i,|X_i|\le \frac{2n}K\right]\\
                h(u^R)&=\phi(u)+\mathds 1\left[u\in X_i,|X_i|\le \frac{2n}K\right]
        \end{align*}

        For any edge $e=(u,v)$ in $G$ that is not any of the barrier edge $e_i$ from $A$ to each $X_i$ and $e_i'$ from each $X_i$ to $A'$, and not in a small fragment $X_i$, the charge $d_e$ can be written as the difference of vertex heights in $G_1$, i.e. $d_e=h(u^L)-h(v^R)$.
\end{lemma}

\begin{proof}
Since the edge $e=(u,v)$ is not any of the barrier edges, its charge simplifies to
\[ d_e=\phi(u)-\phi(v)-\sum_{X_i: |X_i|\le \frac {2n}K}\mathds 1\left[e\in \partial X_i\right] .\]
In addition, since the edge is not in a small fragment $X_i$, the expression $\mathds 1\left[e\in \partial X_i\right]$ is equivalent to $\mathds 1\left[u\in X_i\right]+\mathds 1\left[v\in X_i\right]$. We conclude that $d_e=h(u^L)-h(v^R)$.
\end{proof}

Since both barrier edges and edges in a small fragment $X_i$ have charge $0$, we can ignore them in the sampling process. Therefore, if we remove these edges from the graph $G_1$, the problem reduces to sampling edges $(u^L,v^R)$ with probability proportional to their weight times the absolute difference of vertex heights $|h(u^L)-h(v^R)|$. This can be reduced to the problem of sampling edges from a cut $\partial S$ in an undirected graph as follows. 

\begin{lemma}\label{lem:vertex-height-to-cut-sampling}
        Let $G_1$ be any directed graph with vertex heights $h\in\{-1,0,1\}$. Define an undirected graph $G_2$ as follows: For each vertex $u$ in $G_1$, create two vertices $u^0$ and $u^1$ in $G_2$. For each edge $(u,v)$ in $G_1$, add two copies $(u^0,v^0)$ and $(u^1,v^1)$ of this edge in $G_2$. 
        Let $S$ be the following set defined by vertex heights $h$:
        \[S=\{u^0: h(u)\ge 0\}\cup \{u^1: h(u)\ge 1\}.\]
        Then for any edge $(u,v)$ in $G_1$, the number of copies of this edge in the cut $\partial S$ is exactly $|h(u)-h(v)|$.
\end{lemma}

\begin{proof}
For each $i\in\{0,1\}$, the edge $(u^i,v^i)$ is in $\partial S$ if and only if exactly one of $h(u)\ge i$ and $h(v)\ge i$ is true. Since $h(u)$ and $h(v)$ are restricted to $\{-1,0,1\}$, the number of such $i\in\{0,1\}$ satisfying exactly one of $h(u)\ge i$ and $h(v)\ge i$ is exactly $|h(u)-h(v)|$.
\end{proof}

Therefore, we can apply \Cref{lem:sampling-from-cut} to the auxiliary graph $G_2$ to support the sampling process we need. It remains to maintain these auxiliary graphs $G_1$ and $G_2$ with barrier edges and edges in small fragments removed as the regular reachability algorithm runs. We use the following property of the regular reachability algorithm to maintain these graphs efficiently. The detailed analysis of the regular reachability algorithm is deferred to \Cref{app:resumable-rr}.

\begin{observation}\label{obs:regular-reachability-algorithm}
        The regular reachability algorithm declares barriers as follows: It maintains a vertex set $A$ and a set of fragments $X_1,\ldots,X_k$. During the algorithm, it performs the following operations on these sets:

        \begin{itemize}
                \item Add a new vertex to $A$.
                \item Merge some existing fragments and vertices in $A$ into a new fragment, with a barrier edge connecting $A$ and the new fragment. These vertices are removed from $A$.
                \item Declare that the current set $A$ and fragments $X_1,\ldots,X_k$ form a barrier.
        \end{itemize}
\end{observation}

By this observation, the vertex height $h$ of any copy of a vertex $u$ in $G_1$ is changed at most three times during the regular reachability algorithm: (i) when it is added to $A$, (ii) when it is merged into a small fragment, and (iii) when it is merged into a large fragment.
We delete an edge from $G_1$ if it becomes a barrier edge or an edge in a small fragment, and we add it back if it is no longer such an edge. Barrier edges only lead to $O(n)$ edge deletions and additions. 
To delete edges in small fragments, when we merge some fragments or vertices into a small fragment, we delete all edges between them; and when we merge them into a large fragment, we add all edges within each small fragment back. For a small fragment of size $s$ before merging into large fragments, we make $O(s^2)$ edge deletions and additions. Therefore, the total number of edge deletions and additions is at most $O(ns)=O(\frac{n^2}K)$, as each small fragment has size at most $\frac{2n}K$.

Summing these together, we get a total of $O(\frac{n^2}K)$ edge modifications. This results in $n^{2+o(1)}/K$ time for maintaining the auxiliary graphs when finding one regular path. Since we want to augment the flow by $K/2$ in \Cref{lem:finding-regular-path}, this takes $n^{2+o(1)}$ time in total. Barrier toggle operations take $O(n^2\log (2U))\cdot n^{o(1)}=n^{2+o(1)}$ time in total. This completes the proof of \Cref{lem:finding-regular-path}, and thus the proof of \Cref{thm:main}.

\bibliography{refs}

@article{GK04,
  author  = {Goldberg, Andrew V. and Karzanov, Alexander V.},
  title   = {Maximum Skew-Symmetric Flows and Matchings},
  journal = {Mathematical Programming},
  volume  = {100},
  number  = {3},
  pages   = {537--568},
  year    = {2004},
  doi     = {10.1007/s10107-004-0505-z},
  url     = {https://doi.org/10.1007/s10107-004-0505-z}
}

@article{GK96,
  author  = {Goldberg, Andrew V. and Karzanov, Alexander V.},
  title   = {Path Problems in Skew-Symmetric Graphs},
  journal = {Combinatorica},
  volume  = {16},
  number  = {3},
  pages   = {353--382},
  year    = {1996},
  doi     = {10.1007/BF01261321},
  url     = {https://doi.org/10.1007/BF01261321}
}

@inproceedings{GRST21,
  author    = {Goranci, Gramoz and R{\"a}cke, Harald and Saranurak, Thatchaphol and Tan, Zihan},
  title     = {The Expander Hierarchy and its Applications to Dynamic Graph Algorithms},
  booktitle = {Proceedings of the 2021 {ACM-SIAM} Symposium on Discrete Algorithms ({SODA})},
  pages     = {2212--2228},
  publisher = {Society for Industrial and Applied Mathematics},
  year      = {2021},
  doi       = {10.1137/1.9781611976465.132},
  url       = {https://doi.org/10.1137/1.9781611976465.132}
}

@article{Li26,
  title={Balancing Weights, Directed Sparsification, and Augmenting Paths},
  author={Li, Jason},
  journal={arXiv preprint arXiv:2604.14633},
  year={2026}
}

@misc{HLLZ26,
  author        = {Hair, Isaac M. and Li, George Z. and Li, Jason and Zhang, Junkai},
  title         = {{Bellman--Ford} in Almost-Linear Time},
  year          = {2026},
  eprint        = {2607.19346},
  archivePrefix = {arXiv},
  primaryClass  = {cs.DS},
  doi           = {10.48550/arXiv.2607.19346},
  url           = {https://arxiv.org/abs/2607.19346},
  note          = {Preprint}
}

@inproceedings{ALPS23,
  author    = {Abboud, Amir and Li, Jason and Panigrahi, Debmalya and Saranurak, Thatchaphol},
  title     = {All-Pairs Max-Flow is no Harder than Single-Pair Max-Flow: {Gomory--Hu} Trees in Almost-Linear Time},
  booktitle = {2023 {IEEE} 64th Annual Symposium on Foundations of Computer Science ({FOCS})},
  pages     = {2204--2212},
  publisher = {IEEE},
  year      = {2023},
  doi       = {10.1109/FOCS57990.2023.00137},
  url       = {https://doi.org/10.1109/FOCS57990.2023.00137}
}

@inproceedings{HLRW24,
  author    = {Henzinger, Monika and Li, Jason and Rao, Satish and Wang, Di},
  title     = {Deterministic Near-Linear Time Minimum Cut in Weighted Graphs},
  booktitle = {Proceedings of the 2024 Annual {ACM-SIAM} Symposium on Discrete Algorithms ({SODA})},
  pages     = {3089--3139},
  publisher = {Society for Industrial and Applied Mathematics},
  year      = {2024},
  doi       = {10.1137/1.9781611977912.111},
  url       = {https://doi.org/10.1137/1.9781611977912.111}
}

@inproceedings{CKLPGS22,
  author    = {Chen, Li and Kyng, Rasmus and Liu, Yang P. and Peng, Richard and Probst Gutenberg, Maximilian and Sachdeva, Sushant},
  title     = {Maximum Flow and Minimum-Cost Flow in Almost-Linear Time},
  booktitle = {2022 {IEEE} 63rd Annual Symposium on Foundations of Computer Science ({FOCS})},
  pages     = {612--623},
  publisher = {IEEE},
  year      = {2022},
  doi       = {10.1109/FOCS54457.2022.00064},
  url       = {https://doi.org/10.1109/FOCS54457.2022.00064}
}

@inproceedings{BBLST25,
  author    = {Bernstein, Aaron and Blikstad, Joakim and Li, Jason and Saranurak, Thatchaphol and Tu, Ta-Wei},
  title     = {Combinatorial Maximum Flow via Weighted Push-Relabel on Shortcut Graphs},
  booktitle = {2025 {IEEE} 66th Annual Symposium on Foundations of Computer Science ({FOCS})},
  pages     = {464--485},
  publisher = {IEEE},
  year      = {2025},
  doi       = {10.1109/FOCS63196.2025.00026},
  url       = {https://doi.org/10.1109/FOCS63196.2025.00026}
}

@inproceedings{BNW22,
  author    = {Bernstein, Aaron and Nanongkai, Danupon and Wulff-Nilsen, Christian},
  title     = {Negative-Weight Single-Source Shortest Paths in Near-linear Time},
  booktitle = {2022 {IEEE} 63rd Annual Symposium on Foundations of Computer Science ({FOCS})},
  pages     = {600--611},
  publisher = {IEEE},
  year      = {2022},
  doi       = {10.1109/FOCS54457.2022.00063},
  url       = {https://doi.org/10.1109/FOCS54457.2022.00063}
}

@misc{HJS26b,
  author        = {Haeupler, Bernhard and Jiang, Yonggang and Saranurak, Thatchaphol},
  title         = {Combinatorial Minimum Cost Flow in Almost-Linear Time on Dense Graphs},
  year          = {2026},
  eprint        = {2610.32556},
  archivePrefix = {arXiv},
  primaryClass  = {cs.DS},
  doi           = {10.48550/arXiv.2610.32556},
  url           = {https://arxiv.org/abs/2610.32556},
  note          = {To appear in FOCS 2026}
}

@misc{FLSW26,
  author        = {Fleischmann, Henry and Li, Jason and Saranurak, Thatchaphol and Wang, Benyu},
  title         = {Directed Global Minimum Cut in Almost-Linear Time},
  year          = {2026},
  eprint        = {2610.10898},
  archivePrefix = {arXiv},
  primaryClass  = {cs.DS},
  doi           = {10.48550/arXiv.2610.10898},
  url           = {https://arxiv.org/abs/2610.10898},
  note          = {Preprint}
}

@inproceedings{DMMSY25,
  author    = {Duan, Ran and Mao, Jiayi and Mao, Xiao and Shu, Xinkai and Yin, Longhui},
  title     = {Breaking the Sorting Barrier for Directed Single-Source Shortest Paths},
  booktitle = {Proceedings of the 57th Annual {ACM} Symposium on Theory of Computing ({STOC})},
  pages     = {36--44},
  publisher = {Association for Computing Machinery},
  year      = {2025},
  doi       = {10.1145/3717823.3718179},
  url       = {https://doi.org/10.1145/3717823.3718179}
}

@inproceedings{MV80,
  author    = {Micali, Silvio and Vazirani, Vijay V.},
  title     = {An {$O(\sqrt{|V|}\,|E|)$} Algorithm for Finding Maximum Matching in General Graphs},
  booktitle = {21st Annual Symposium on Foundations of Computer Science ({FOCS})},
  pages     = {17--27},
  publisher = {IEEE Computer Society},
  year      = {1980},
  doi       = {10.1109/SFCS.1980.12},
  url       = {https://doi.org/10.1109/SFCS.1980.12}
}

@inproceedings{MS04,
  author    = {Mucha, Marcin and Sankowski, Piotr},
  title     = {Maximum Matchings via {Gaussian} Elimination},
  booktitle = {45th Annual {IEEE} Symposium on Foundations of Computer Science ({FOCS})},
  pages     = {248--255},
  publisher = {IEEE Computer Society},
  year      = {2004},
  doi       = {10.1109/FOCS.2004.40},
  url       = {https://doi.org/10.1109/FOCS.2004.40}
}
\bibliographystyle{alpha}

\appendix

\section{Resumable Regular Reachability}\label{app:resumable-rr}

For a barrier $(A,A',X_1,\ldots,X_k,M)$ as in \Cref{def:barrier}, the unique edge $e_i$ from $A$ to $X_i$ and its symmetric edge $e_i'$ (the unique edge from $X_i$ to $A'$) are the \emph{barrier edges}. A symmetric pair $\{e,e'\}$ of new edges \emph{violates} a barrier of $H$ if the barrier is not a barrier of $H\cup\{e,e'\}$. For instance, every edge with negative charge in \Cref{def:operation-vector} violates the barrier that defines the charges.

Linear-time algorithms that find a regular $s$-$s'$ path or a barrier are given in~\cite{GK96} (see also \cite[Section~7]{GK04}): a search structure is grown from $s$, and self-symmetric fragments are trimmed when they are discovered. These algorithms are stated for a fixed graph, while we need to stop at each barrier and continue after inserting edges. This is essentially free because in an Edmonds-type search the edges may be scanned in \emph{any} order: an inserted edge is just one more edge to scan, and no earlier step of the search has to be undone. Below we describe the algorithm in this form. Its analysis closely follows~\cite{GK96,GK04}, so we only sketch it.

\begin{lemma}\label{lem:resumable-rr}
    Given a source vertex $s$ (along with its mate $s'$) and an initial skew-symmetric graph $H$, there is an algorithm which runs until it either outputs a regular $s$-$s'$ path in $H$ or pauses. When it is paused, symmetric edge pairs can be inserted into $H$, and then the algorithm resumes. The algorithm has the following guarantees:
	%There is an algorithm that maintains a partition $V=A\sqcup A'\sqcup X_1\sqcup\cdots\sqcup X_k\sqcup M$, where $A'$ is the symmetric set of $A$ and the \emph{fragments} $X_i$ and $M$ are self-symmetric, together with an edge $e_i\in E$ from $A$ to $X_i$ for every $i$. Started on a graph $H$, it runs until it either outputs a regular $s$-$s'$ path in $H$ or \emph{pauses}. While it is paused, symmetric edge pairs may be inserted into $H$, and then it is resumed. The algorithm guarantees:
	\begin{enumerate}[label=(\roman*)]
		\item\label{it:rr-barrier} \textbf{Pauses.} Whenever the algorithm pauses, it has found a barrier $(A,A',X_1,\ldots,X_k,M)$ of the current $H$ and $A\cup X_1\cup\cdots\cup X_k$ is the set of vertices reachable from $s$ by regular paths in $H$. While symmetric edge pairs are inserted, the algorithm takes $O(1)$ time per pair and only resumes the search when some inserted pair violates the current barrier.
		\item\label{it:rr-monotone} \textbf{Progress.} The algorithm maintains a partition $(A,A',X_1,\ldots,X_k,M)$ which initially has $A=\{s\}$, $A'=\{s'\}$ and $k=0$. Between each resume and pause, the partition changes only through the following two operations:
		\begin{itemize}
			\item \emph{Grow:} a pair $\{v,v'\}\subseteq M$ leaves $M$; $v$ is added to $A$ and $v'$ to $A'$.
			\item \emph{Merge:} for a nonempty set $S\subseteq A$ and fragments $X_{i_1},\ldots,X_{i_r}$ ($r\ge 0$), the sets $S,S',X_{i_1},\ldots,X_{i_r}$ are replaced by the single fragment $S\cup S'\cup X_{i_1}\cup\cdots\cup X_{i_r}$, which receives a barrier edge from $A$.
		\end{itemize}
		%\item\label{it:rr-progress} %\textbf{Progress.} If none of the pairs inserted during a pause violates the current barrier, the algorithm pauses again immediately, with the same barrier, after $O(1)$ time per inserted pair. Otherwise, at least one operation of~\ref{it:rr-monotone} happens before the algorithm pauses again or outputs a path.
	\end{enumerate}
    The total running time is $\tilde O(m)$, where $m$ is the final number of edges of $H$.
\end{lemma}

The algorithm for static regular reachability depends on the following contraction operation.

\begin{theorem}[Theorem 2.3 of \cite{GK96}]
	Suppose that there is a regular path $P$ from some vertex $v$ to its symmetric vertex $v'$, and there is another regular path $Q$ from $s$ to $v$ that does not intersect $P$ except at $v$. Let $e_v$ be the last edge of $Q$ incident to $v$. Define the contraction operation as follows:
	\begin{itemize}
		\item Remove all vertices on $P$ and the symmetric path $P'$ except for $v$ and $v'$. Replace the path with a pair of symmetric edge from $v$ to $v'$.
		\item For any edge whose destination is on $P$ or $P'$ (except edge $e_v$), redirect it to $v'$.
		\item For any edge whose source is on $P$ or $P'$ (except edge $e_v'$), redirect it to $v$.
	\end{itemize} 
	This contraction operation preserves regular reachability from $s$ to $s'$, and a regular path from $s$ to $s'$ in the original graph can be reconstructed from a regular path in the contracted graph efficiently.
\end{theorem}
%proceeds as follows.
The algorithm works as follows. We maintain a set $S$ of vertices which we already know can be reached from $s$ via a regular path, and a tree $T$ that represents all regular paths from $s$ to the vertices in $S$. We run a breadth-first search like algorithm to explore new vertices. If we find a regular path $P_1$ from $s$ to a vertex $v$, and another regular path $P_2$ from $s$ to its symmetric vertex $v'$, we can take the path $P_1\circ P_2'$, which either gives a regular path from $s$ to $s'$ directly, or gives a regular path from some vertex $u$ to its symmetric vertex $u'$ that contains $v$. In the latter case, we apply the contraction operation to this path, which contracts both $v$ and $v'$. Using these steps, $S$ never contains both a vertex and its symmetric vertex simultaneously, except for these contracted pairs at the boundary. When the algorithm cannot make further progress, all edges exiting $S$ are from those contracted pairs at the boundary, which corresponds to the fragments in the barrier. The remaining part of $S$ corresponds to the set $A$ in the barrier. The property of maintaining fragments that we want comes from the contraction operation. Resuming the algorithm after adding new edges is essentially the same as continuing the breadth-first search.

\end{document}